\documentclass[journal, twoside, onecolumn, draftclsnofoot, 12pt]{IEEEtran}
\usepackage{graphicx}
\usepackage{verbatim}

\ifCLASSINFOpdf
\else
\fi
\usepackage[cmex10]{amsmath}
\newtheorem{lemma}{Lemma}[section]
\newtheorem{proposition}{Proposition}[section]
\newtheorem{corollary}{Corollary}[section]

\newenvironment{proof}[1][Proof]{\begin{trivlist}
\item[\hskip \labelsep {\bfseries #1}]}{\end{trivlist}}

\newcommand{\qed}{\nobreak \ifvmode \relax \else
      \ifdim\lastskip<1.5em \hskip-\lastskip
      \hskip1.5em plus0em minus0.5em \fi \nobreak
      \vrule height0.75em width0.5em depth0.25em\fi}

\begin{document}
\title{Random Deployment of Data Collectors \\ for Serving Randomly-Located Sensors}

\author{Taesoo Kwon,~\IEEEmembership{Member,~IEEE}
        %John~Doe,~\IEEEmembership{Fellow,~OSA,}
        and~John~M.~Cioffi,~\IEEEmembership{Fellow,~IEEE}% <-this % stops a space
\thanks{%Manuscript received XXX XX, 2011; The associate editor
%coordinating the review of this paper and approving it for
%publication was XXX.

T. Kwon and J. M. Cioffi are with the Department of Electrical
Engineering, Stanford University, Stanford
CA, 94305, USA (e-mail: tskwon80@stanford.edu, cioffi@stanford.edu)}% <-this % stops a space
%\thanks{Digital Object Identifier XXX}
}

\maketitle

%%% Need to revise
% Section IV: in equation (11), the effect of alpha on required density of data collectors
% Section IV: in equation (15) and (16), delete the subscripts of alpha, 'o' and 'm'
% Section V: C(1,alpha')/C(1,alpha) holds when only beta_t = 0 dB

\begin{abstract}
%\boldmath

Recently, wireless communication industries have begun to extend
their services to machine-type communication devices as well as to
user equipments. Such machine-type communication devices as meters
and sensors need intermittent uplink resources to report measured
or sensed data to their serving data collector. It is however hard
to dedicate limited uplink resources to each of them. Thus,
efficient service of a tremendous number of devices with low
activities may consider simple random access as a solution. The
data collectors receiving the measured data from many sensors
simultaneously can successfully decode only signals with
signal-to-interference-plus-noise-ratio (SINR) above a certain
value. The main design issues for this environment become how many
data collectors are needed, how much power sensor nodes transmit
with, and how wireless channels affect the performance. This paper
provides answers to those questions through a stochastic analysis
based on a spatial point process and on simulations.

\end{abstract}


% \fi
\IEEEpeerreviewmaketitle

%ts \IEEEPARstart{T}{his} demo file is intended to serve as a ``starter file''
%ts for IEEE journal papers produced under \LaTeX\ using IEEEtran.cls
%ts version 1.7 and later.
%ts I wish you the best of success.
%ts \hfill mds
%ts \hfill January 11, 2007

\section{Introduction} \label{SEC_intro}

Wireless personal communication enables ubiquitous exchange of
various data types such as voice, video, photos, and text among
individuals. The emergence of new advanced systems such as the
IEEE 802.11ac \cite{REF_IEEE802_11} and the 3GPP LTE-Advanced
\cite{REF_3GPP} are expected to achieve additional data rates. Of
late, wireless communication industries have begun to discuss
their scenarios serving machine-type communication devices such as
meters/sensors as well as user equipments such as smart phones
\cite{REF_3GPP_M2M_Service}\cite{REF_3GPP_M2M_System}. These
machine-to-machine (M2M) communications have extensive
applications, from monitoring environments to full
electrical/mechanical automation (e.g. smart grid, smart city,
Internet of things), which has been being considered as one of the
most crucial technologies in future
\cite{REF_M2M_Chen}\cite{REF_M2M_Wu}.
% Random single-hop M2M network
The sensor network can also be regarded as a kind of M2M, and
there have been many studies in the form of ad-hoc networks
\cite{REF_Sensor_Akyildiz}\cite{REF_Sensor_Kulkarni}. This paper
only considers the environment with specific data collectors
directly communicating with sensors. This environment is suitable
when sensor nodes support only simple single-hop communication
functionalities and deployment of many data collectors is easy.
This type of M2M communication is similar to cellular
communication systems where the base stations serve user equipment
within their coverage, but it has the unique characteristics
\cite{REF_M2M_Chen}\cite{REF_ETSI_M2M}: there can be a huge number
of devices (e.g. trillions) each of which has only a small amount
of data and a low activity, and their functionalities have to be
simple. These characteristics may require technologies
differentiated from the conventional high data rate human-to-human
(H2H) communications. For example, machine-type devices such as
meters and sensors need uplink resources intermittently for
reporting measured or sensed data to their serving data collector,
but it is hard to dedicate limited uplink resources to each. Thus,
simple random access can be considered as a solution for directly
transmitting measured data or initially requesting uplink
resources. The data collectors that receive many sensors' measured
data simultaneously can successfully decode only signals with
signal-to-interference-plus-noise ratio (SINR) above a certain
value. In order to keep a high success probability of many sensor
nodes' intermittent transmissions, the system may need a lot of
data collectors, and conventional macro/micro base stations may
not be appropriate for these roles. In other words, data
collectors have to be easy to deploy and cost-effective. They
support only simple functionalities and are interconnected with
external networks through wired or wireless links. It can be
considered that not only a new type of device for data collection
is defined but also such devices as pico/femto base stations
around sensor nodes play the role of data collectors.
Fig.~\ref{FIG_System_Model} shows a system architecture with data
collectors and sensor nodes. In this environment, some questions
are: How many data collectors are needed? How much transmit power
sensors have to use for successful transmission? And, how the
wireless channels affect the performance. This paper will provide
answers to those questions through a stochastic analysis based on
a spatial point process and on simulations.

\begin{comment}

\begin{figure}[t]
\centering
\includegraphics[width=7cm]{./fig/SystemModel.eps}
\caption{Data collectors to collect data from sensor nodes}
\label{FIG_System_Model}
\end{figure}

\end{comment}

% Stochastic geometry

The main factor of determining system performance is the
interference from neighbor sensor nodes. This interference depends
on the spatial distribution and sensor-node access methods.
Because the spatial configurations of transmitting and receiving
nodes can have enormous possibilities, it is impossible to
consider each possibility. Stochastic geometry provides a useful
mathematical tool to model network topology, and it also enables
analysis of essential quantities such as interference distribution
and outage
\cite{REF_StochasticGeometry_Haenggi}\cite{REF_StochasticGeometryNow_Baccelli}\cite{REF_SpatialModelSurvey_Andrews}.
This stochastic geometry has mainly been applied to pure ad hoc
networks and their performance has been analyzed under the
assumption of random transmitter location and receiver with fixed
distances to its transmitter
\cite{REF_StochasticGeometry_Haenggi}\cite{REF_AdhocExp_Andrews}\cite{REF_OppAloha_Baccelli}.
This paper considers the environment where both transmitters
(sensor nodes) and receivers (data collectors) are randomly
deployed and each transmitter are served by the data collector
nearest to it.
\cite{REF_CellOld1_Baccelli}-\cite{REF_CellNakagami_Andrews} have
analyzed the distribution of signal-to-interference ratio (SIR) or
SINR in random cellular networks where both transmitter and
receiver are randomly located;
\cite{REF_ShotgunCellular_Madhusudhanan} analyzed the distribution
of SIR considering the path loss and shadowing,
\cite{REF_CellRayleigh_Andrews} derived a simple-form SINR
distribution in case of Rayleigh fading and a path-loss exponent
of four, and \cite{REF_CellNakagami_Andrews} expanded the analysis
results in \cite{REF_CellRayleigh_Andrews} into the results for a
more general fading model including Nakagami-$m$ fading. But, they
assumed that each base station always has the user equipment
within its coverage and communicates with a user equipment that is
scheduled exclusively within one cell and focused on
transmitter-centric coverage (i.e. downlink).
\cite{REF_UplinkCDMA_Mehta} modeled CDMA uplink interference power
as a log-normal distribution using the moment-matching method.
Also, \cite{REF_UplinkMIMO_Staelin} asymptotically analyzed uplink
spectral efficiency in spatially distributed wireless networks,
where the base stations have multiple antennas, by using
infinite-random-matrix theory and stochastic geometry. The current
system is similar to the uplink cellular systems, but this paper
will only consider random access without any explicit scheduling
for the simple functionalities of sensor nodes and data
collectors.

%This analysis considers SIR or SINR distribution in case of
%Nakagami-$m$ fading and Rayleigh fading.

The three contributions of this paper are: First, an analysis
shows how the channel affects the SIR distribution for
Nakagami-$m$ fading. A simple form on the SINR distribution is
found for some special channel models. Second, an analysis
describes how many data collectors per area are on average
required to meet the outage probability for the given mean number
of sensor nodes per area in case of Rayleigh fading. Third, this
paper suggests a simple design method of the transmit power and
the mean number of data collectors to meet the given outage
probability.

The remainder of this paper is organized as follows:
Section~\ref{SEC_Model} presents the system model based on a
homogeneous Poisson point process (PPP). Section~\ref{SEC_Sinr}
analyzes the SIR distribution for Nakagammi-$m$ fading channels
and the SINR distribution for Rayleigh fading channels.
Section~\ref{SEC_Density} derives the intensity of data collectors
required to keep the outage probability below a certain value and
suggests a design method of the transmit power.
Section~\ref{SEC_Results} discusses numerical results. Finally,
Section~\ref{SEC_Conclusions} concludes.

% Contribution

% Organization

\section{System Model} \label{SEC_Model}

A sensor node senses or measures environments and then transmits
its data to the closest data collector. Sensor nodes do not always
have data to transmit but send them only when their sensing data
are generated. For example, machines such as meters and event
sensors may transmit data intermittently rather than continuously,
and it has to be successful with probability above a certain
value. In order to model intermittent transmissions, the sensor
node's activity is defined as $\rho$. This value of $\rho$ is
between $0$ and $1$, and this paper considers its small values.
Meanwhile, data collectors that receive data from sensor nodes,
are always ready to receive data from them.

%These data collectors may be deployed by some operators, moreover,

\begin{comment}

\begin{figure}[t]
\centering
\includegraphics[width=9cm]{./fig/VoronoiTessellation.eps}
\caption{Data collectors and sensor nodes distributed by a
homogeneous Poisson point processes. Data collectors build the
Voronoi tessellation ($\lambda_{s,total} = 10^{-2}\,{\rm m}^{-2}$,
$\lambda_c = 5 \times 10^{-3}\,{\rm m}^{-2}$, $\rho = 0.01$)}
\label{FIG_VoroniTessellation}
\end{figure}

\end{comment}

This paper considers environments where both of sensor nodes and
data collectors are randomly deployed. Sensor nodes are
distributed according to a homogeneous PPP, $\Phi_s$, and they
transmit sensed data to their nearest data collectors through
random access schemes. $\lambda_{s,total}$ denotes the intensity
of sensor nodes that is the average number of them per area. In
order to consider unplanned deployments of data collectors, the
random locations of data collectors are modeled as a homogeneous
PPP, $\Phi_c$, with intensity $\lambda_c$, like sensor nodes. Each
sensor node transmits its data to a data collector closest to it,
so a data collector builds a coverage based on Voronoi
tessellation, as shown in Fig~\ref{FIG_VoroniTessellation}.

The standard power loss propagation model with the path loss
exponent $\alpha(>2)$ and the Nakagami-$m$ fading model are
considered. In the Nakagami-$m$ fading model
\cite{REF_NakagamiFading}, $m=1$, $m=(K+1)^2/(2K+1)$ and
$m=\infty$ model Rayleigh fading, Rician fading with parameter
$K$, and no fading, respectively. Also, it is assumed that all
sensor nodes transmit with the same power $P$. A typical data
collector located on the origin receives the signal with
$Pr^{-\alpha}G_S$ from a typical sensor node when the distance
between them is $r$ and the fading channel gain is $G_S$. By
Slyvnyak's theorem \cite{REF_StochasticGeometry_Book}, interfering
nodes except for a typical sensor node located on $X_0$ still
constitute a homogeneous PPP with intensity $\lambda_{s,total}$.
Thus, the interference power of the link between a typical sensor
node and a typical data collector can be expressed as $I_r =
\sum_{X_j \in \Phi_s \backslash \{X_0\}}
P\left|X_j\right|^{-\alpha}G_{I,j}$ where $X_j$ denotes the
location of a interfering node and $G_{I,j}$ means the fading gain
of a link between a typical data collector and a interfering
sensor node $j$. $\left\{G_{I,j}\right\}_{X_j \in \Phi_s
\backslash \{X_0\}}$ are independently and identically distributed
(i.i.d.) random variables. Here, it is assumed that a typical data
collector does not perform any scheduling for sensor nodes within
coverage served by itself, so they may interfere with each other
even though they are served by a common data collector.
Eventually, the link of a transmitter-receiver pair experiences
interference from interfering nodes distributed according to a
homogeneous PPP with effective intensity $\lambda_s =
\lambda_{s,total} \cdot \rho$. If a sensor transmits using one of
$N$ resources that is chosen at random, $\lambda_s$ is
$\lambda_{s,total} \cdot \rho /N$. $\lambda_s$ decreases as $N$
increases and this means that $N$ is also a parameter for the
system design. In this paper, $N=1$ is assumed.

When the interference is dealt with as noise and single antenna is
equipped on both transmitters and receivers, the SINR is given by
\begin{equation}\label{EQN_SinrModel}
    \begin{array}{ll}
        \mathrm{SINR} & = \frac{P\left|X_0\right|^{-\alpha}G_S}
        {\sum_{X_j \in \Phi_s \backslash \{X_0\}}
        P\left|X_j\right|^{-\alpha}G_{I,j} + \sigma^2} \\
        & = \frac{\left|X_0\right|^{-\alpha}G_S}
        {\sum_{X_j \in \Phi_s \backslash \{X_0\}}
        \left|X_j\right|^{-\alpha}G_{I,j} + \tilde{\sigma}^2}
    \end{array}
\end{equation}
where $\sigma^2$ is the noise power and $\tilde{\sigma}^2$ is
equal to $\frac{\sigma^2}{P}$. In case of $\sigma^2 \rightarrow
0$, (\ref{EQN_SinrModel}) means the SIR.

\section{SINR Distribution} \label{SEC_Sinr}

This section analyzes the SINR distributions and derive
simple-form SIR or SINR distribution for some specific channels.
For more generalization of results, the fading gain distribution
of (\ref{EQN_GeneralFading}) is first considered.
\begin{equation}\label{EQN_GeneralFading}
    \Pr\left\{G_S>g\right\}=\sum_{n \in \mathcal{N}}
    \exp(-ng) \sum_{k \in \mathcal{K}} a_{nk}g^k
\end{equation}
for some finite set $\mathcal{N}$ and a finite integer set
$\mathcal{K}$. This type of complementary cumulative distribution
function (CCDF) includes a variety of fading-gain distributions
such as exponential distribution, chi-square distribution and
gamma distribution.

\begin{lemma}\label{LEM_Sinr_GneralFading}

    \emph{Let sensor nodes and data collectors distributed with
    homogeneous PPP's with intensities $\lambda_s$ and $\lambda_c$,
    respectively and each sensor node builds communication link
    with a data collectors closest to it. When the CCDF of the fading gain of a desired signal
    is given by (\ref{EQN_GeneralFading}) and the fading gain of the interfering signal is denoted as a random variable,
    $G_{I}$, the CCDF of SINR is given by}

    \begin{equation}\label{EQN_Sinr_GeneralFading}
        \begin{array}{ll}
            \Pr\left\{{\rm SINR} > \beta \right\} \\ =
            2\pi \lambda_c \sum_{n \in \mathcal{N}}\sum_{k \in
            \mathcal{K}} a_{nk} \left(-\beta \right)^k  \\
            \int_{0}^{\infty} r^{k \alpha+1}
            \left.\frac{d^k \exp\left(-\lambda_s \xi(\zeta,\alpha)-\zeta \tilde{\sigma}^2\right)}
            {d\zeta^k} \right|_{\zeta=n \beta r^\alpha} \exp\left(-\lambda_c \pi r^2 \right) dr
        \end{array}
    \end{equation}
    \emph{where $\xi (\zeta,\alpha) = \pi \zeta^{\frac{2}{\alpha}}
    \Gamma\left(1-\frac{2}{\alpha}\right)
    \mathrm{E}\{G_I^{\frac{2}{\alpha}}\}$, $\mathrm{E}\{x\}$ is the expectation
    of $x$ and $\Gamma(x)=\int_{0}^{\infty} t^{x-1}\exp(-t)dt$ denotes
    the gamma function. The derivative in (\ref{EQN_Sinr_GeneralFading}) can be reexpressed as follow.}

    \begin{equation}\label{EQN_Sinr_GeneralFading_Derivative}
        \begin{array}{ll}
            \frac{d^k \exp\left(-\lambda_s \xi(\zeta,\alpha)-\zeta \tilde{\sigma}^2\right)}
            {d\zeta^k} = \\
            \exp \left(-\lambda_s \xi(\zeta,\alpha) - \zeta \tilde{\sigma}^2\right) \sum_{l=0}^{k}
            \frac{1}{l!}\sum_{j=0}^{l}(-1)^{l+j}\binom{l}{j} \\
            \hspace{1cm}\left[\lambda_s\xi(\zeta,\alpha)+\zeta \tilde{\sigma}^2
            \right]^j \frac{\partial^k}{\partial \zeta^k} \left[\lambda_s \xi(\zeta,\alpha)+\zeta
            \tilde{\sigma}^2 \right]^{l-j}
        \end{array}
    \end{equation}

\end{lemma}

\begin{proof}
See Appendix~\ref{APP_LEM_Sinr_GneralFading}.
\end{proof}

The result of SINR distribution in
Lemma~\ref{LEM_Sinr_GneralFading} requires cumbersome integrations
and differentiations, but simple-form result can be obtained for
specific channel models.

To begin with, an analysis considers Nakagami-$m$ fading channel.
The received signal power experiencing Nakagami-$m$ fading channel
can be modeled using Gamma distributions. Thus, assuming that
desired signals experience Nakagami-$m_s$ fading while interfering
signals experience Nakagami-$m_i$ fading, the CCDF of fading gains
can be give by
\begin{equation}\label{EQN_NakagamiFading_S}
    \Pr\left\{G_S>g\right\}=\sum_{k=0}^{m_s-1} \frac{(m_s g)^k}{k!}
    \exp(-m_s g)
\end{equation}
\begin{equation}\label{EQN_NakagamiFading_I}
    \Pr\left\{G_I>g\right\}=\sum_{k=0}^{m_i-1} \frac{(m_i g)^k}{k!}
    \exp(-m_i g)
\end{equation}
(\ref{EQN_NakagamiFading_S}) and (\ref{EQN_NakagamiFading_I}) have
the forms of (\ref{EQN_GeneralFading}), so the SINR distribution
can be derived by using Lemma~\ref{LEM_Sinr_GneralFading}.
Generally, Lemma~\ref{LEM_Sinr_GneralFading} requires the
calculation of a derivative in
(\ref{EQN_Sinr_GeneralFading_Derivative}) and it is too complex to
calculate it for any $m_s$, $\alpha$ and $\tilde{\sigma}$.
Fortunately, it is possible to obtain a simple form for the CCDF
of SINR under interference limited environments, i.e. $\sigma^2
\rightarrow 0$.

\begin{proposition}\label{PRO_Sir_Nakagami}

    \emph{Let sensor nodes be randomly located with intensity
    $\lambda_s$ and served by the nearest data collectors randomly deployed with
    intensity intensity $\lambda_c$. When their links experience Nakagami-$m$ fading given
    by (\ref{EQN_NakagamiFading_S}) and (\ref{EQN_NakagamiFading_I}),
    and $\tilde{\sigma}\rightarrow 0$, the CCDF of SIR
    is given by}

    \begin{equation}\label{EQN_Sir_Nakagami}
        \begin{array}{ll}
            \Pr\left\{{\rm SIR} > \beta \right\}  =
            \frac{\lambda_c}{\lambda_c + \lambda_s C(m_i, \alpha) (m_s\beta)^{\frac{2}{\alpha}}
            } \cdot
            \\ \sum_{k=0}^{m_s-1} \frac{1}{k!} \sum_{l=0}^{k}
            (-1)^{l+k} \Delta_{k,l}
            \left[\frac{\lambda_s C(m_i, \alpha)
            (m_s\beta)^{\frac{2}{\alpha}}}{\lambda_c + \lambda_s C(m_i, \alpha)
            (m_s\beta)^{\frac{2}{\alpha}}}
            \right]^l
        \end{array}
    \end{equation}

    \emph{where
    $C(m,\alpha)= \frac{m^{-\frac{2}{\alpha}} \Gamma\left(1-\frac{2}{\alpha}\right)
    \Gamma\left(m+\frac{2}{\alpha}\right)}{\Gamma\left(m\right)}$ and
    $\Delta_{k,l} = \sum_{j=0}^{l}(-1)^j \binom{l}{j} \prod_{i=0}^{k-1} \left[\frac{2}{\alpha}(l-j)-i\right]$ for $l \leq k$.
    Here, $\Delta_{0,0}$ is defined as $1$.}

\end{proposition}
\begin{proof}
    See Appendix~\ref{APP_PRO_Sir_Nakagami}.
\end{proof}

In Nakagami-$m$ fading model, $m=1$ means the Rayleigh fading
model. Thus, it is also easy to obtain the CCDF of the SIR for
Rayleigh fading model.

\begin{corollary}\label{COR_Sir_Rayleigh}

    \emph{Let sensor nodes be randomly located with intensity
    $\lambda_s$ and served by the nearest data collector randomly deployed with
    intensity intensity $\lambda_c$. When all links experience Rayleigh fading
    with unit mean, and $\tilde{\sigma}\rightarrow 0$, the CCDF of SIR
    is given by}

    \begin{equation}\label{EQN_Sir_Rayleigh}
        \begin{array}{ll}
            \Pr\left\{{\rm SIR} > \beta \right\}  = \frac{\lambda_c}{\lambda_c + \lambda_s C(1,\alpha) \beta^{\frac{2}{\alpha}}}
        \end{array}
    \end{equation}

    \emph{where
    $C(1,\alpha)=\Gamma\left(1-\frac{2}{\alpha}\right)\Gamma\left(1+\frac{2}{\alpha}\right)=\frac{2\pi}{\alpha\sin(2\pi/\alpha)}$.}

\end{corollary}
\begin{proof}
    By substituting $m_s=1$ and $m_i=1$ into the results in
    Proposition~\ref{PRO_Sir_Nakagami}, (\ref{EQN_Sir_Rayleigh}) is obtained. Also, $C(1,\alpha)$ can be calculated by using the property of the gamma
    function $\Gamma(1-z)\Gamma(z)=\frac{\pi}{\sin(\pi z)}$.
\end{proof}

As discussed before, when the noise power cannot be neglected, it
is hard to obtain a simple form of the SINR for general $m_s$
because it requires the derivative of
(\ref{EQN_Sinr_GeneralFading_Derivative}). However, when a path
loss exponent, $\alpha$, is equal to $4$ and the fading channel is
modeled as Rayleigh fading, the CCDF of SINR is simplified into a
common integral form.

\begin{proposition}\label{PRO_Sinr_Rayleigh}

    \emph{Let sensor nodes be randomly located with intensity
    $\lambda_s$ and served by the nearest data collector randomly deployed with
    intensity intensity $\lambda_c$. when all links experience Rayleigh fading
    with unit mean and a path loss exponent $\alpha$ is $4$, the CCDF of
    SINR is given by}

    \begin{equation}\label{EQN_Sinr_Rayleigh}
        \begin{array}{ll}
            \Pr\left\{{\rm SINR} > \beta \right\} \\ =
            \frac{\pi^{\frac{3}{2}}\lambda_c}{2\sqrt{\beta \tilde{\sigma}^2}}
            \exp\left( \frac{[\pi \lambda_c + K \beta^{\frac{1}{2}} \lambda_s ]^2}{4 \beta \tilde{\sigma}^2} \right)
            {\rm erfc}\left(\frac{\pi \lambda_c + K \beta^{\frac{1}{2}} \lambda_s}{2 \sqrt{\beta \tilde{\sigma}^2}}\right)
        \end{array}
    \end{equation}

    \emph{where
    $K=\frac{\pi^2}{2}$ and
    ${\rm erfc}(x) = \frac{2}{\sqrt{\pi}}\int_{x}^{\infty}
    \exp\left(-t^2\right)dt$ is the complementary error
    function.}

\end{proposition}

\begin{proof}
    The CCDF of fading gain for $m_s=1$
    is the case of $\mathcal{N}=\left\{1\right\}$, $\mathcal{K}=\left\{0\right\}$ and
    $a_{10}=1$ in (\ref{EQN_GeneralFading}). Thus, when $\alpha=4$
    \begin{equation}\label{EQN_Sinr_Rayleigh_Proof}
        \begin{array}{ll}
            \Pr\left\{{\rm SINR}>\beta\right\}  \\
            = 2 \pi \lambda_c
            \int_{0}^{\infty} r \cdot \exp \left(- \pi r^2 \left[\lambda_s \beta^{\frac{1}{2}} C(1,4) +
            \lambda_c\right]-\beta r^{4} \tilde{\sigma}^2\right) dr
            %\\ \stackrel{(b)}{=} \frac{\pi^{\frac{3}{2}}\lambda_c}{2\sqrt{\beta \tilde{\sigma}^2}}
            %\exp\left( \frac{[\pi \lambda_c + K \beta^{\frac{1}{2}} \lambda_s ]^2}{4 \beta \tilde{\sigma}^2} \right)
            %{\rm erfc}\left(\frac{\pi \lambda_c + K \beta^{\frac{1}{2}} \lambda_s}{2 \sqrt{\beta
            %\tilde{\sigma}^2}}\right)
        \end{array}
    \end{equation}
    (\ref{EQN_Sinr_Rayleigh_Proof}) follows from (\ref{EQN_Sinr_GeneralFading}) and
    (\ref{EQN_Xi_Sir_Nakagami}), and it can be evaluated by using the change of variables $r^2 \rightarrow x$ and
    the integration formula, $\int_{0}^{\infty}
    \exp\left(-\left[ax+bx^2\right]\right)dx = \frac{1}{2 \sqrt{b}}
    \exp\left(\frac{a^2}{4b}\right){\rm erfc}\left(\frac{a}{2 \sqrt{b}}\right)$ for $a\geq0$ and
    $b>0$. Here, $K=\pi C(1,4) = \frac{\pi^2}{2}$.
\end{proof}

Proposition~\ref{PRO_Sinr_Rayleigh} is the result for $\alpha=4$.
When $\alpha$ is not $4$, the CCDF of SINR can be expressed by
generalized hypergeometric functions. But they are not simple, so
this paper does not deal with them.

\section{Intensity of Data Collectors} \label{SEC_Density}

When sensor nodes are spatially distributed according to a
homogeneous PPP with a certain intensity, it is important to
decide how many data collectors should be deployed in order to
keep the success probability of random accesses above a certain
value. This section analyzes the requirement of the intensity of
data collectors deployed at random, and the effect of channels on
its required intensity, given the intensity of sensor nodes and a
target outage probability. The outage probability, $\varepsilon$,
is defined as $\Pr\{{\rm SINR < \beta_t}\}$ where $\beta_t$ is the
minimal SINR value required for the successful receptions.

The required intensity of data collectors for Rayleigh fading is
presented in Corollary~\ref{COR_Density_Rayleigh_Sir} and
Proposition~\ref{PRO_Density_Rayleigh_Sinr}.

\begin{corollary}\label{COR_Density_Rayleigh_Sir}

    \emph{Let sensor nodes randomly located with intensity $\lambda_s$
    and served by the nearest data collectors. It is assumed that all
    links experience Rayleigh fading with unit mean and
    $\tilde{\sigma}\rightarrow 0$. The necessary and sufficient condition of the intensity of data
    collectors randomly deployed, $\lambda_c$, for keeping the outage
    probability below $\varepsilon_t$, is}
    \begin{equation}\label{EQN_Density_Rayleigh_Sir}
        \begin{array}{ll}
            \lambda_c \geq \frac{1}{\varepsilon_t}(1-\varepsilon_t)C(1,\alpha)\beta_t^{\frac{2}{\alpha}}\lambda_s
        \end{array}
    \end{equation}
    \emph{where $C(1,\alpha)$ is defined in
    Corollary~\ref{COR_Sir_Rayleigh}.}

\end{corollary}

\begin{proof}
(\ref{EQN_Density_Rayleigh_Sir}) can be directly derived from
(\ref{EQN_Sir_Rayleigh}).
\end{proof}

\begin{proposition}\label{PRO_Density_Rayleigh_Sinr}

    \emph{Let sensor nodes randomly located with intensity $\lambda_s$
    and served by the nearest data collectors. It is assumed that all
    links experience Rayleigh fading with unit mean and $\alpha=4$.
    The sufficient condition of the intensity of data
    collectors randomly deployed, $\lambda_c$, for keeping the outage
    probability below $\varepsilon_t$, is}
    \begin{equation}\label{EQN_Density_Rayleigh_Sinr}
        \begin{array}{ll}
            \lambda_c \geq \frac{K \beta_t^{\frac{1}{2}}}{2\pi
            \varepsilon_t}\left[(1-2\varepsilon_t) + \sqrt{1 +
            8\varepsilon_t(1-\varepsilon_t)\frac{\tilde{\sigma}^2}{(K\lambda_s)^2}}\right]\lambda_s
        \end{array}
    \end{equation}
    \emph{where $K=\frac{\pi^2}{2}$.}

\end{proposition}

\begin{proof}
See Appendix~\ref{APP_PRO_Density_Rayleigh_Sinr}.
\end{proof}

The condition of $\lambda_c$ in (\ref{EQN_Density_Rayleigh_Sir})
under interference-limited environments is necessary and
sufficient while the condition in
(\ref{EQN_Density_Rayleigh_Sinr}) under environments with
non-neglectable noise is just sufficient. In fact,
(\ref{EQN_Density_Rayleigh_Sinr}) has been derived from a lower
bound of the complementary error function. But, for small values
of $\tilde{\sigma}^2$, (\ref{EQN_Density_Rayleigh_Sinr}) also
gives a tight lower bound of $\lambda_c$, and in particular,
(\ref{EQN_Density_Rayleigh_Sinr}) is the same as
(\ref{EQN_Density_Rayleigh_Sir}) with $\alpha=4$ when
$\tilde{\sigma}^2\rightarrow 0$.

Here, given $\beta_t$, $\varepsilon_t$, $\sigma^2$ and
$\lambda_s$, the design method of the transmit power ($P$) of
sensor nodes and the intensity ($\lambda_c$) of data collectors is
suggested, for a path loss exponent of four. The relations among
these variables are given by (\ref{EQN_Sinr_Rayleigh}), but it is
not easy to use (\ref{EQN_Sinr_Rayleigh}) directly for the design
of $P$ and $\lambda_c$. On the contrary, the lower bound of the
CCDF of SINR with a simpler form of
(\ref{EQN_PROOF_PRO_Density_Rayleigh_Sinr_LB}) can give a simple
design method for them. The lower bound of SINR CCDF in
(\ref{EQN_PROOF_PRO_Density_Rayleigh_Sinr_LB}) is equivalent to
the intensity condition of data collectors of
(\ref{EQN_Density_Rayleigh_Sinr}) and the second term within a
square root in (\ref{EQN_Density_Rayleigh_Sinr}) approximately
models the effect of noise. Now, the transmit power and the
intensity of data collectors can be separately designed. First,
for neglecting the noise effect, the second term within the square
root of (\ref{EQN_Density_Rayleigh_Sinr}) has to be much smaller
that $1$. The definition of $\tilde{\sigma}^2$ gives a condition
of the transmit power.
\begin{equation}\label{EQN_TxPw_Condition}
    \begin{array}{ll}
        8\varepsilon_t (1-\varepsilon_t) \frac{\sigma^2/P}{(K
        \lambda_s)^2} \stackrel{\mathrm{(a)}}{\leq}
        \frac{8\sigma^2/P}{(\pi^2
        \lambda_s)^2} \ll 1
    \end{array}
\end{equation}
where (a) follows from the inequality of arithmetic and geometric
means and the definition of $K$. Thus, the transmit power can be
set to
\begin{equation}\label{EQN_TxPw_Design}
    \begin{array}{ll}
        P = c \cdot \frac{8 \sigma^2}{\pi^4\lambda_s^2}
    \end{array}
\end{equation}
where $c$ is a constant much less than one and it is a design
parameter. $c$ has to be set not only to neglect the noise power
but also to keep transmit power as small as possible for sensor
node's power saving. Next, the intensity of data collectors can be
designed according to (\ref{EQN_Density_Rayleigh_Sir}) because the
intensity condition (\ref{EQN_Density_Rayleigh_Sinr}) is almost
equal to (\ref{EQN_Density_Rayleigh_Sir}) if $P$ is set by
(\ref{EQN_TxPw_Design}). In this design, the transmit power is
reciprocally proportional to $\lambda_s^2$ and this means that the
longer the distance among sensor nodes is, the larger the required
transmit power is, because of the noise effect, when the intensity
of data collector is determined by
(\ref{EQN_Density_Rayleigh_Sir}). Even though this design method
is very simple, but it gives a good design method for the random
deployment of data collectors to serve randomly distributed
wireless sensors. Its performances will be shown in
Section~\ref{SEC_Results}.

In interference-limited environments with Rayleigh fading
channels, Corollary~\ref{COR_Density_Rayleigh_Sir} shows the
effect of the path loss exponents obviously. Because the function
$\frac{x}{\sin x}$ is a increasing function of $0<x<\pi$, it is
obvious that the required density of data collectors decreases as
the path loss exponent increases, when $\alpha > 2$ and $\beta_t
\geq 1$ for given $\varepsilon_t$ and $\lambda_s$, from the
definition of $C(1,\alpha)$ and (\ref{EQN_Density_Rayleigh_Sir}).
%\begin{equation}\label{EQN_Effect_PathLossExp}
%    \begin{array}{ll}
%        \frac{\lambda_c (\alpha_1)}{\lambda_c (\alpha_2)} =
%        \frac{C(1,\alpha_1)}{C(1,\alpha_2)} =
%        %\frac{\Gamma\left(1-\frac{2}{\alpha_1}\right)\Gamma\left(1+\frac{2}{\alpha_1}\right)}
%        %{\Gamma\left(1-\frac{2}{\alpha_2}\right)\Gamma\left(1+\frac{2}{\alpha_2}\right)}
%        \frac{\alpha_2 \sin(2\pi/\alpha_2)}{\alpha_1 \sin(2\pi/\alpha_1)}
%    \end{array}
%\end{equation}
%where $\lambda_c (\alpha)$ denotes the required intensity of data
%collectors for the path loss exponent of $\alpha$.

On the contrary, it is not easy to express the required intensity
of date collectors in case of the Nakagami-$m$ fading with general
$m$'s, in a simple form. Here, the effect of wireless channels on
system designs is analyzed by comparing the performances to those
of Rayleigh fading, rather than deriving their requirements
exactly, only when $\tilde{\sigma}^2\rightarrow 0$. When deploying
data collectors with intensity $\lambda_{c,o}$ for randomly
distributed sensor nodes with intensity $\lambda_s$, let
$\varepsilon^{(o)}$ and $\varepsilon^{(m)}$ denote the outage
probabilities for the Rayleigh fading model and the another
examined-fading model for the required SIR $\beta_t$,
respectively. First, in case of reference channel model assuming
the Rayleigh fading, the $\lambda_{c,o}$ and $\varepsilon^{(o)}$
have the following relation from (\ref{EQN_Sinr_Rayleigh}).
\begin{equation}\label{EQN_Density_Effect_Rayleigh}
    \begin{array}{ll}
        \lambda_{c,o} = \frac{1}{
        \varepsilon^{(o)}}(1-\varepsilon^{(o)})C(1,\alpha)\beta_t^{\frac{2}{\alpha}}\lambda_s
    \end{array}
\end{equation}
On the other hand, in case of the examined-fading channel,
$\tilde{\lambda}_{c,m}$ is defined as
\begin{equation}\label{EQN_Density_Effect_Nakagami}
    \begin{array}{ll}
        \tilde{\lambda}_{c,m} = \frac{1}{
        \varepsilon^{(m)}}(1-\varepsilon^{(m)})C(1,\alpha)\beta_t^{\frac{2}{\alpha}}\lambda_s
    \end{array}
\end{equation}
where $\varepsilon^{(m)}$ is derived from
(\ref{EQN_Sir_Nakagami}). In other words,
(\ref{EQN_Density_Effect_Nakagami}) means that the deployment of
data collectors with $\lambda_{c,o}$ in the Nakagami-$m$ fading
channel is equal to the deployment of data collectors with
$\tilde{\lambda}_{c,m}$ in the Rayleigh fading channel in term of
outage probability. Hence, $\tilde{\lambda}_{c,m}/\lambda_{c,o}$
quantifies the effect of wireless fading channels on the system
design and is simplified from (\ref{EQN_Density_Effect_Rayleigh})
and (\ref{EQN_Density_Effect_Nakagami}), as follows.
\begin{equation}\label{EQN_Density_Gain}
    \begin{array}{ll}
        \frac{\tilde{\lambda}_{c,m}}{\lambda_{c,o}} =
        \frac{\varepsilon^{(o)}(1-\varepsilon^{(m)})}{\varepsilon^{(m)}(1-\varepsilon^{(o)})}
        = \frac{1/\varepsilon^{(m)}-1}{1/\varepsilon^{(o)}-1}
    \end{array}
\end{equation}

\section{Numerical Results and Discussion} \label{SEC_Results}

This section evaluates and discusses the performance of systems
with data collectors randomly deployed to serve randomly
distributed wireless sensors, based on results of
Section~\ref{SEC_Sinr} and Section~\ref{SEC_Density}. It is
assumed that the total intensity of sensor nodes
($\lambda_{s,total}$) spatially distributed according to a
homogeneous PPP is $10^{-2} {\rm m}^{-2}$. Also, $\rho$ is set to
$10^{-4}$ and it means that the sensor nodes awake on average
every $1000$ sec (about 17 minutes), when they transmit data to
data collectors during $100$ msec on each awake mode. Also, the
minimal SINR value ($\beta_t$) required for the successful
reception of $0$ dB is considered.

Fig.~\ref{FIG_SinrCdf_Rayleigh} shows the CDF of SINR according to
$P$ and $\frac{\lambda_c}{\lambda_s}$. This can be interpreted as
the outage probability for $\beta_t$ which is a value on x-axis.
$P/\sigma^2$'s (or $1/\tilde{\sigma}^2$) of $100$ dB and $120$ dB
are assumed. These values mean that the transmit powers of sensor
nodes are $-10$ dBm ($0.1$ mW) and $10$ dBm ($10$ mW), when the
power spectral density of the noise is $-170$ dBm/Hz and the
bandwidth is $1$ MHz. Fig.~\ref{FIG_SinrCdf_Rayleigh} indicates
that analysis results in (\ref{EQN_Sir_Rayleigh}) and
(\ref{EQN_Sinr_Rayleigh}) definitely coincide with the simulation
results. When $P/\sigma^2$ is $100$ dB, the
$\frac{\lambda_c}{\lambda_s}$'s of $10$ and $20$ result in the
outage probabilities of $0.23$ and $0.1$, respectively. As
$P/\sigma^2$ increases, outage probability decreases. In other
words, larger intensity of data collectors and higher transmit
power lead to less outage probability.
Fig.~\ref{FIG_EffectNoise_Rayleigh} and
Fig.~\ref{FIG_EffectTxPower_Rayleigh} explain these effects more
quantitatively. In Fig.~\ref{FIG_EffectNoise_Rayleigh}, the outage
probability decreases as the intensity of data collectors
increases, and their required intensity can be obtained for a
given outage probability. Also its lower bound by
(\ref{EQN_Density_Rayleigh_Sinr}) is shown. The lower bound of
$\lambda_c$ in (\ref{EQN_Density_Rayleigh_Sinr}) is tighter when
the effect of noise is reduced. The effect of noise on outage
probability decreases as $\frac{\lambda_c}{\lambda_s}$ increases.
This is because the increase of $\frac{\lambda_c}{\lambda_s}$
leads to the increase of received SNR because of the decrease in
distances between data collectors and sensor nodes.
Fig.~\ref{FIG_EffectTxPower_Rayleigh} shows how the transmit power
of sensor nodes affect the outage probability. The results of
Fig.~\ref{FIG_EffectTxPower_Rayleigh} were evaluated by changing
the intensity of sensor nodes for given relative intensities of
data collectors. In other words, it shows the effect of noise by
changing the geometric size of networks. The larger geometric
size, i.e. larger distances between sensor nodes and data
collectors, leads to the bigger effects of noise on the system
performance. These results also verify that the design of the
transmit power not only reduces the noise effect but also keeps
the transmit power as small as possible. Also, under the
environments of Fig.~\ref{FIG_SinrCdf_Rayleigh}, the design by
(\ref{EQN_TxPw_Design}) with $c=0.1$ provides the transmit power
of $9$ dBm (i.e. $P/\sigma^2=119$ dB), and it is observed that
$P/\sigma^2=120$ dB approaches the performance of $P/\sigma^2
\rightarrow \infty$ in Fig.~\ref{FIG_SinrCdf_Rayleigh}. These
results confirm that (\ref{EQN_TxPw_Design}) is a very efficient
design method. The path loss exponent is another crucial factor to
have an effect on system performances. As
Fig.~\ref{FIG_SinrCdf_Rayleigh_PathlossExp} indicates, they result
in very different performance for the same transmit power. At
$P/\sigma^2 = 100$ dB, the noise can be neglected in case of a
pathloss exponent $3$ while it causes severe performance
degradation in case of a pathloss exponent $5$. By contrast, when
the noise effect can be neglected, larger $\alpha$'s result in
less outage probability for given $\beta_t$ and $\lambda_s$. In
fact, for given $\varepsilon_t$ and $\lambda_s$, the required
$\lambda_c$ for a pathloss exponent of $\alpha'$ increases by the
factor of $\frac{C(1,\alpha')}{C(1,\alpha)}
\beta_t^{(\frac{2}{\alpha'}-\frac{2}{\alpha})}$, compared to a
pathloss exponent of $\alpha$ when $P/\sigma^2 \rightarrow 0$,
where $C(1,\alpha)$ is defined in
Corollary~\ref{COR_Density_Rayleigh_Sir}. For example, when
$\beta_t = 0~\rm{dB}$, the path loss exponents of $3$ and $5$
requires $1.54$ and $0.84$ times of the intensity of data
collectors for the path loss exponent of $4$.

Fig.~\ref{FIG_SinrCdf_Nakagami} -
Fig.~\ref{FIG_EffectChannel_Density} examine the performance for
Nakagami-$m$ fading channels. Fig.~\ref{FIG_SinrCdf_Nakagami}
explains how the line-of-sight factors of fading channels
contribute to the SINR distribution. The increase in $m$ results
in the decrease in outage probability. But, $m$ more than two does
not have an big effect on the performance, compared to $m$ equal
to two. Fig.~\ref{FIG_SinrCdf_Nakagami} also indicates that
analysis results exactly coincide with simulation results when
considering that the performance of $P/\sigma^2 = 120$ dB is as
good as that of $P/\sigma^2 \rightarrow \infty$.
Fig.~\ref{FIG_EffectChannel_Nakagami} shows the effect of channels
on outage probability under the interference-limited environments.
The outage probability decreases as $m$ and the pathloss exponent
increase. It means that the Rician fading and AWGN environments
need less intensity of data collectors than the Rayleigh fading
environments for the same path loss exponent. Moreover, from this
figure, the intensity of data collectors required to meet a
certain outage probability can be obtained.
Fig.~\ref{FIG_EffectChannel_Density} examines the relative effect
of other fading channels compared to the Rayleigh fading channel
in term of the intensity of data collectors, which is defined in
(\ref{EQN_Density_Gain}). it shows that $m$ and $\alpha$ has a big
effect on the system design such as the deployment of data
collectors.

%Especially, the larger $\frac{\lambda_c}{\lambda_s}$ is, the
%bigger the impact of wireless channels on the system performance
%is.

So far, this paper analyzed and discussed the effect of the
wireless channels, the transmit power and the intensity of data
collectors on system performances when data collectors are
randomly deployed to successfully collect the data from
randomly-located sensor nodes. As the number of wireless nodes
increases enormously in future, it is more and more difficult to
design the system. For reducing these difficulties, efficient
system design methods is required to deal with a huge number of
wireless nodes, so the rigorous understanding about the spatial
distribution and effect of interference will be basics for them.
Even though this paper has considered only simple random access,
these results will be able to be used as basic models for
developing more sophisticated spatial resource management methods.

\section{Conclusions} \label{SEC_Conclusions}

This paper has considered the environment where receivers (data
collectors) as well as transmitters (sensor nodes) are randomly
deployed and each transmitter is served by the receiver nearest to
it. In network topology modeled by homogeneous Poisson point
processes, analysis and simulation results showed the SINR
distribution, and a simple design method of transmit power was
suggested. Under interference-limited environments, the larger the
path loss exponent and the portion of line-of-site factors were,
the less the outage probability was. On the contrary, under
non-neglectable noise environments, the large path loss exponent
caused severe performance degradation. Moreover, the intensity of
data collectors required to keep the outage probability above a
certain value was derived, and it depends on required outage
probability, an intensity of sensor nodes, a fading channel model,
a path loss exponent and noise power. This required intensity
helps to design such parameters as the amount of wireless
resources and the access probability for medium access control.
Random access scheme is very simple and does not cause
control-overhead problems even under environments with a huge
number of sensor nodes, but its required intensity of data
collectors is never small. Thus, it is needed to find more
sophisticated spatial resource management schemes and the result
of this paper may be used as a basic model for them.

\appendices

\section{Proof of Lemma~\ref{LEM_Sinr_GneralFading}} \label{APP_LEM_Sinr_GneralFading}

This proof is similar to proof of theorem 1 in
\cite{REF_CellNakagami_Andrews} that has considered the
transmitter-centric coverage (or downlink) and only the
transmitter intensity. Here, an analysis focuses on the
receiver-centric coverage by data collectors (or uplink) and
allows that multiple transmitters within the service area of a
common data collector simultaneously transmit. For those
differences and the completeness, this paper provides the full
derivation of the CCDF of SINR.

The probability that there is a data collector at a distance of
$r$ from a typical sensor node is $2\pi \lambda_c dr$. For this
data collector to be a serving data collector of a typical sensor
node, all other data collectors must be farther than $r$ from a
typical sensor node, and its probability is $\exp(-\lambda_c \pi
r^2)$. Thus, the probability density function of the distance
between a typical sensor node and its serving data collector,
$f_r(r)$, is equal to $2 \pi \lambda_c r \cdot \exp(-\lambda_c \pi
r^2)$.

The CCDF of SINR is
\begin{equation}\label{EQN_PROOF_LEM11_SINR}
    \begin{array}{ll}
        \Pr\{{\rm SINR}>\beta\} \\
        = \int_{0}^{\infty}
        \Pr\left\{\frac{r^{-\alpha}G_S}
        {I_r + \tilde{\sigma}^2} >
        \beta\right\} f_r(r) dr \\
        = 2\pi\lambda_c\int_{0}^{\infty} \Pr\left\{G_S >
        \beta r^{\alpha} (I_r + \tilde{\sigma}^2)\right\} r \exp(-\lambda_c \pi
        r^2) dr
    \end{array}
\end{equation}
where $I_r = \sum_{X_j \in \Phi_s \backslash \{X_0\}}
\left|X_j\right|^{-\alpha}G_{I,j}$. From
(\ref{EQN_GeneralFading}),
\begin{equation}\label{EQN_PROOF_LEM11_Gs_CCDF}
    \begin{array}{ll}
        \Pr\left\{G_S > \beta r^{\alpha} (I_r +
        \tilde{\sigma}^2)\right\} \\
        = \mathrm{E}_{I_r}\left\{ \sum_{n \in \mathcal{N}}
        \exp(-n\beta r^\alpha [I_r + \tilde{\sigma}^2]) \cdot \right. \\
        \hspace{2.5cm} \left. \sum_{k \in \mathcal{K}} a_{nk}(\beta r^\alpha [I_r + \tilde{\sigma}^2])^k
        \right\} \\
        = \sum_{n \in \mathcal{N}}\sum_{k \in \mathcal{K}} a_{nk}
        (\beta r^\alpha)^k \cdot \\
        \hspace{2.5cm}\mathrm{E}_{I_r}\left\{ (I_r + \tilde{\sigma}^2)^k \exp(-n\beta r^\alpha [I_r + \tilde{\sigma}^2])
        \right\} \\
        \stackrel{\mathrm{(a)}}{=} \sum_{n \in \mathcal{N}}\sum_{k \in \mathcal{K}} a_{nk}
        (-\beta r^\alpha)^k \left.\frac{d^k \mathrm{E}\{\exp\left(-\zeta(I_r+\tilde{\sigma}^2)\right)\}}
        {d\zeta^k} \right|_{\zeta=n \beta r^\alpha} \\
        \stackrel{\mathrm{(b)}}{=} \sum_{n \in \mathcal{N}}\sum_{k \in \mathcal{K}} a_{nk}
        (-\beta r^\alpha)^k \left.\frac{d^k \mathcal{L}_{I_r} (\zeta) \exp(-\zeta \tilde{\sigma}^2)}
        {d\zeta^k} \right|_{\zeta=n \beta r^\alpha}
    \end{array}
\end{equation}
where (a) and (b) follow from the definition of Laplace transform,
$\mathcal{L}_{X}(\zeta)=\mathrm{E}_{X}\{\exp(-\zeta X)\}$, its
property, $\mathcal{L}_{t^k X(t)}(\zeta)=(-1)^k \frac{d^k
\mathcal{L}_{X}(\zeta)}{d \zeta^k}$, and the independence of $I_r$
and $\tilde{\sigma}^2$. The Laplace transform of $I_r$ is
\begin{equation}\label{EQN_PROOF_LEM11_L_Ir}
    \begin{array}{ll}
        \mathcal{L}_{I_r}(\zeta)& =\mathrm{E}_{I_r}\{\exp(-\zeta
        I_r)\} \\
        & = \mathrm{E}_{\Phi_s,G_{I}}\{ \exp(-\zeta \sum_{X_j \in \Phi_s \backslash \{X_0\}}
        \left|X_j\right|^{-\alpha}G_{I,j}) \} \\
        & = \mathrm{E}_{\Phi_s}\{\prod_{X_j \in \Phi_s \backslash \{X_0\}}
        \mathrm{E}_{G_{I,j}}\{ \exp(-\zeta G_{I,j} \left|X_j\right|^{-\alpha}) \}
        \} \\
        & \stackrel{\mathrm{(c)}}{=} \exp\left( -2\pi\lambda_s \int_{0}^{\infty} [1-
        \mathrm{E}_{G_I}\{ \exp(-\zeta G_{I} v^{-\alpha} )\}] v dv
        \right) \\
        & \stackrel{\mathrm{(d)}}{=} \exp\left( -2\pi\lambda_s \cdot \right. \\
        & \hspace{1.3cm} \left. \int_{0}^{\infty} \left(\int_{0}^{\infty}[1-
        \exp(-\zeta v^{-\alpha} g] v dv \right) f_{G_I}(g) dg
        \right) \\
        & \stackrel{\mathrm{(e)}}{=} \exp\left( -\frac{2\pi\lambda_s
        \zeta^{\frac{2}{\alpha}}}{\alpha} \Gamma\left(-\frac{2}{\alpha}\right)
        \int_{0}^{\infty} g^{\frac{2}{\alpha}} f_{G_I}(g) dg
        \right) \\
        & \stackrel{\mathrm{(f)}}{=} \exp\left( -\lambda_s \xi(\zeta,\alpha) \right)
    \end{array}
\end{equation}
where (c) follows from the probability generating functional
(PGFL) of the PPP \cite{REF_StochasticGeometry_Book}; (d) uses the
probability density function $f_{G_I}(g)$ of a random variable
$G_I$; (e) follows from the change of variable $v^{-\alpha}
\rightarrow x$ and the definition of the Gamma function; (f)
follows from the property of Gamma function $x
\Gamma(x)=\Gamma(1+x)$ and the definition of $\xi(\zeta,\alpha)$.

By substituting (\ref{EQN_PROOF_LEM11_Gs_CCDF}) and
(\ref{EQN_PROOF_LEM11_L_Ir}) into (\ref{EQN_PROOF_LEM11_SINR}),
(\ref{EQN_Sinr_GeneralFading}) is derived.

Also, (\ref{EQN_Sinr_GeneralFading_Derivative}) is obtained from
the following equation which can be derived by the derivative of
the exponential function and the chain rules:
\begin{equation}\label{EQN_ExpComposite_Derivative}
    \begin{array}{ll}
        \frac{\partial^k} {\partial z^k} \exp\left(f(z)\right) =
        \\ \exp \left(f(z)\right) \sum_{l=0}^{k}
        \frac{1}{l!}\sum_{j=0}^{l}(-1)^{j}\binom{l}{j} f(z)^j \frac{\partial^k f(z)^{l-j}}{\partial z^k}
    \end{array}
\end{equation}
where $\binom{l}{j}$ denotes $\frac{l!}{j!(l-j)!}$.

\section{Proof of Proposition~\ref{PRO_Sir_Nakagami}} \label{APP_PRO_Sir_Nakagami}

The fading gain of Nakagami-$m$ fading channel given in
(\ref{EQN_NakagamiFading_S}) can be reexpressed as
\begin{equation}\label{EQN_NakagamiFading_S_Reform}
    \Pr\left\{G_S>g\right\}=\sum_{n=m_s}^{m_s}\exp(-n g)\sum_{k=0}^{m_s-1}
    \frac{n^k}{k!}g^k
\end{equation}
So, the Nakagami-$m$ fading is the case of
$\mathcal{N}=\left\{m_s\right\}$,
$\mathcal{K}=\left\{0,\cdots,m_s-1\right\}$ and $a_{nk} =
\frac{n^k}{k!}$. Thus,
\begin{equation}\label{EQN_Xi_Sir_Nakagami}
        \begin{array}{ll}
            \xi\left(\zeta, \alpha \right) & = \pi \zeta^{\frac{2}{\alpha}} \Gamma \left(1-\frac{2}{\alpha}\right)
            \int_{0}^{\infty} g^{\frac{2}{\alpha}}\cdot \frac{g^{m_i-1}\exp(-m_i g)}{m_i^{-m_i}\Gamma(m_i)}dg \\
            & = \pi \zeta^{\frac{2}{\alpha}} C(m_i,\alpha)
        \end{array}
    \end{equation}
where $C(m,\alpha)$ is defined as
$\frac{m^{-\frac{2}{\alpha}} \Gamma\left(1-\frac{2}{\alpha}\right)
\Gamma\left(m+\frac{2}{\alpha}\right)}{\Gamma\left(m\right)}$.

When $\tilde{\sigma}^2 \rightarrow 0$, the derivative
(\ref{EQN_Sinr_GeneralFading_Derivative}) is calculated into
\begin{equation}\label{EQN_Sir_Nakagami_Derivative}
    \begin{array}{ll}
        \frac{d^k \exp\left(-\lambda_s \xi(\zeta,\alpha)\right)}
        {d\zeta^k} \\
        = \exp \left(-\lambda_s \pi \zeta^{\frac{2}{\alpha}} C(m_i,\alpha) \right) \sum_{l=0}^{k}
        \frac{1}{l!}\sum_{j=0}^{l}(-1)^{l+j}\binom{l}{j} \\
        \hspace{1cm}\left[\lambda_s \pi \zeta^{\frac{2}{\alpha}} C(m_i,\alpha)
        \right]^j \frac{\partial^k}{\partial \zeta^k} \left[\lambda_s \pi \zeta^{\frac{2}{\alpha}} C(m_i,\alpha)
        \right]^{l-j} \\
        = \exp \left(-\lambda_s \pi \zeta^{\frac{2}{\alpha}} C(m_i,\alpha) \right) \sum_{l=0}^{k}
        \frac{1}{l!}\sum_{j=0}^{l}(-1)^{l+j}\binom{l}{j} \\
        \hspace{1cm} [-\lambda_s \pi C(m_i,\alpha)]^{l}
        \left[ \prod_{i=0}^{k-1} \left(\frac{2}{\alpha}(l-j)-i\right)
        \right] \zeta^{\frac{2}{\alpha}l-k}
    \end{array}
\end{equation}
From (\ref{EQN_NakagamiFading_S_Reform}) and
(\ref{EQN_Sir_Nakagami_Derivative}),
(\ref{EQN_Sinr_GeneralFading}) is
\begin{equation}\label{EQN_Sir_Nakagami_Proof}
    \begin{array}{ll}
        \Pr\left\{{\rm SIR}>\beta\right\}  \\
        \stackrel{\mathrm{(a)}}{=} 2 \pi \lambda_c \sum_{k=0}^{m_s-1}
        \frac{m_s^k}{k!}(-\beta)^k \\ \hspace{1cm} \int_{0}^{\infty}
        (m_s\beta)^{-k} \sum_{l=0}^{k}
        \frac{(-1)^l}{l!}\left[\lambda_s \pi
        C(m_i,\alpha)(m_s\beta)^{\frac{2}{\alpha}}r^2\right]^{l}
        \cdot \\ \hspace{1cm} \Delta_{k,l} \cdot \exp\left(-[\lambda_s \pi C(m_i,\alpha)(m_s\beta)^{\frac{2}{\alpha}} + \lambda_c
        \pi]r^2\right) r dr \\
        \stackrel{\mathrm{(b)}}{=} 2 \pi \lambda_c \sum_{k=0}^{m_s-1}
        \frac{(-1)^k}{k!} \\ \hspace{1cm} \sum_{l=0}^{k}
        \frac{(-1)^l}{l!} \left[ \lambda_s \pi C(m_i,\alpha)(m_s\beta)^{\frac{2}{\alpha}}
        \right]^l \Delta_{k,l} \\ \hspace{1cm} \int_{0}^{\infty}
        \exp\left( -[\lambda_s \pi C(m_i,\alpha)(m_s\beta)^{\frac{2}{\alpha}} + \lambda_c
        \pi]r^2 \right)r^{2l+1}dr \\
        \stackrel{\mathrm{(c)}}{=} 2 \pi \lambda_c \sum_{k=0}^{m_s-1}
        \frac{(-1)^k}{k!} \\ \hspace{1cm} \sum_{l=0}^{k}
        \frac{(-1)^l}{l!} \left[ \lambda_s \pi C(m_i,\alpha)(m_s\beta)^{\frac{2}{\alpha}}
        \right]^l \Delta_{k,l} \\ \hspace{1cm} \left( \frac{1}{2} \left[ \lambda_s \pi
        C(m_i,\alpha)(m_s\beta)^{\frac{2}{\alpha}}+\lambda_c \pi
        \right]^{-l-1} \Gamma(l+1)\right) \\
        \stackrel{\mathrm{(d)}}{=} \frac{\lambda_c}{\lambda_c + \lambda_s C(m_i, \alpha) (m_s\beta)^{\frac{2}{\alpha}}} \sum_{k=0}^{m_s-1} \frac{1}{k!} \cdot \\ \hspace{1cm}
        \sum_{l=0}^{k}
        (-1)^{k+l} \Delta_{k,l} \left[\frac{\lambda_s C(m_i, \alpha)
        (m_s\beta)^{\frac{2}{\alpha}}}{\lambda_c + \lambda_s C(m_i, \alpha)
        (m_s\beta)^{\frac{2}{\alpha}}}
        \right]^l
    \end{array}
\end{equation}
where (a) follows from the definition of $\Delta_{k,l}$ in
Proposition~\ref{PRO_Sir_Nakagami}, (b) follows from the
interchange of a summation and an integration, (c) follows from
the calculation of the integral part by the definition of the
Gamma function, and (d) follows from the property of the Gamma
function $\Gamma(l+1)=l!$ for a nonnegative integer $l$.

\section{Proof of Proposition~\ref{PRO_Density_Rayleigh_Sinr}} \label{APP_PRO_Density_Rayleigh_Sinr}

Let $\tau=\frac{\pi \lambda_c + K \beta_t^{\frac{1}{2}}
\lambda_s}{2 \sqrt{\beta_t\tilde{\sigma}^2}}$ and $\kappa =
\pi^{\frac{3}{2}}\frac{\lambda_c}{2\sqrt{\beta_t
\tilde{\sigma^2}}}$. (\ref{EQN_Sinr_Rayleigh}) can be expressed as
\begin{equation}\label{EQN_PROOF_PRO_Density_Rayleigh_Sinr_ErfcLB}
    \begin{array}{ll}
    \Pr\left\{{\rm SINR} > \beta \right\}
    = \exp(\tau^2){\rm erfc}(\tau)\kappa \\
    \stackrel{\rm(a)}{>} \exp(\tau^2) \cdot
    \frac{2}{\sqrt{\pi}}\frac{\tau}{1+2\tau^2}\exp(-\tau^2) \cdot
    \kappa \\
    = \frac{2}{\sqrt{\pi}}\frac{\tau}{1+2\tau^2} \cdot \kappa
    \end{array}
\end{equation}
where (a) follows from the lower bound of the complementary error
function. From (\ref{EQN_Sinr_Rayleigh}) and
(\ref{EQN_PROOF_PRO_Density_Rayleigh_Sinr_ErfcLB}),
\begin{equation}\label{EQN_PROOF_PRO_Density_Rayleigh_Sinr_LB}
    \begin{array}{ll}
    \Pr\left\{{\rm SINR} > \beta \right\}
    & \stackrel{\rm(b)}{>} \pi\lambda_c \cdot \frac{\pi \lambda_c + K \beta_t^{\frac{1}{2}}\lambda_s}
    {2\beta_t \tilde{\sigma}^2 + (\pi \lambda_c + K
    \beta_t^{\frac{1}{2}}\lambda_s)^2} \\
    & \geq 1-\varepsilon_t
    \end{array}
\end{equation}
where (b) follows from the definition of $\tau$ and $\kappa$.
(\ref{EQN_PROOF_PRO_Density_Rayleigh_Sinr_LB}) is rewritten into
\begin{equation}\label{EQN_PROOF_PRO_Density_Rayleigh_Sinr_QuadEq}
    \begin{array}{ll}
    \varepsilon_t \pi^2 \lambda_c^2 - [(1-2\epsilon_t) \pi K \beta_t^{\frac{1}{2}}\lambda_s]
    \lambda_c\\
    \hspace{2cm} - (1-\varepsilon_t)(K^2 \beta_t \lambda_s^2 + 2
    \beta_t
    \tilde{\sigma}^2) \geq 0
    \end{array}
\end{equation}
which is a quadratic inequality with the form of
$a\lambda_c^2+b\lambda_c+c \geq 0$ where $a>0$ and $c<0$ for
$0<\varepsilon_t<1$. Thus,
(\ref{EQN_PROOF_PRO_Density_Rayleigh_Sinr_QuadEq}) gives a
positive lower bound of $\lambda_c$. By solving the inequality
(\ref{EQN_PROOF_PRO_Density_Rayleigh_Sinr_QuadEq}) for a variable
$\lambda_c>0$, (\ref{EQN_Density_Rayleigh_Sinr}) is derived.

\appendices
%ts \section{Proof of the First Zonklar Equation}
%ts Appendix one text goes here.

%ts \section{}
%ts Appendix two text goes here.

%ts \section*{Acknowledgment}

%\ifCLASSOPTIONcaptionsoff
%ts   \newpage
%\fi



%\bibliographystyle{IEEEtran}

%\begin{comment}
\clearpage
\begin{figure}[p]
\centering
\includegraphics[width=12cm]{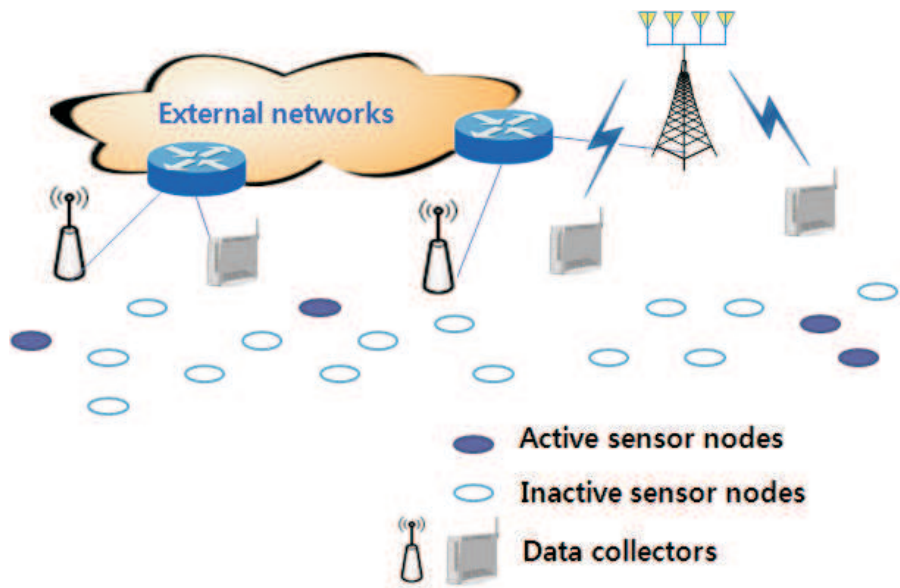}
\caption{Data collectors to collect data from sensor nodes}
\label{FIG_System_Model}
\end{figure}

\begin{figure}[p]
\centering
\includegraphics[width=12cm]{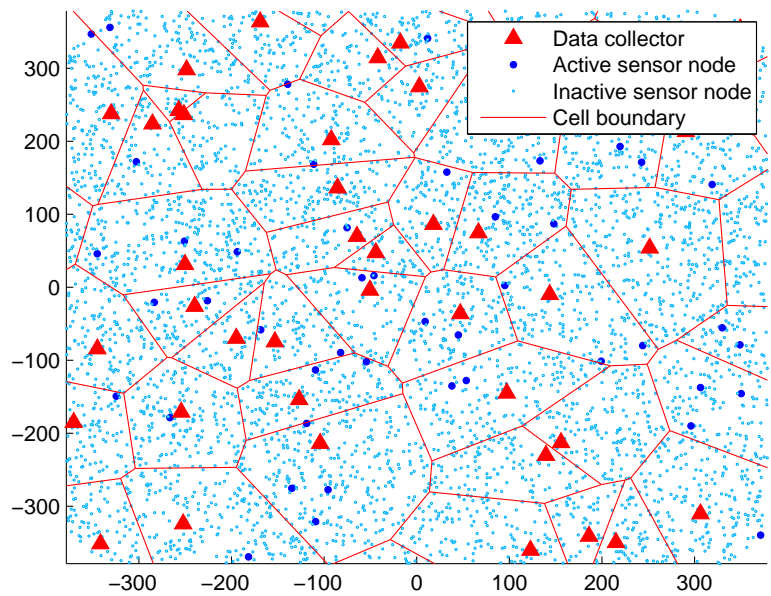}
\caption{Data collectors and sensor nodes distributed by a
homogeneous Poisson point processes. Data collectors build the
Voronoi tessellation ($\lambda_{s,total} = 10^{-2}\,{\rm m}^{-2}$,
$\lambda_c = 5 \times 10^{-3}\,{\rm m}^{-2}$, $\rho = 0.01$)}
\label{FIG_VoroniTessellation}
\end{figure}

\begin{figure}[p]
\centering
\includegraphics[width=12cm]{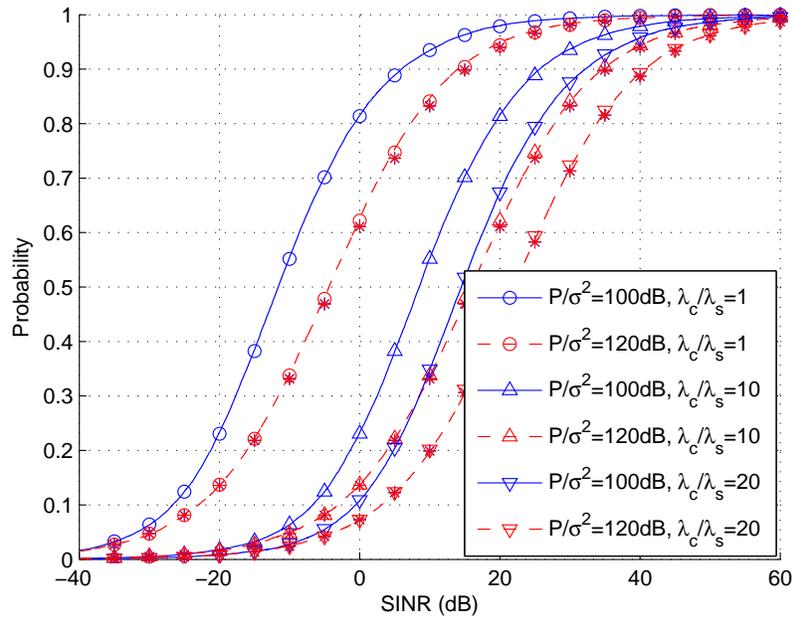}
\caption{CDF of SINR for Rayleigh fading channels ($\alpha = 4$;
lines - simulation results; symbols - analysis results; star
symbols represent the case of $P/\sigma^2 \rightarrow \infty$)}
\label{FIG_SinrCdf_Rayleigh}
\end{figure}

\begin{figure}[p]
\centering
\includegraphics[width=12cm]{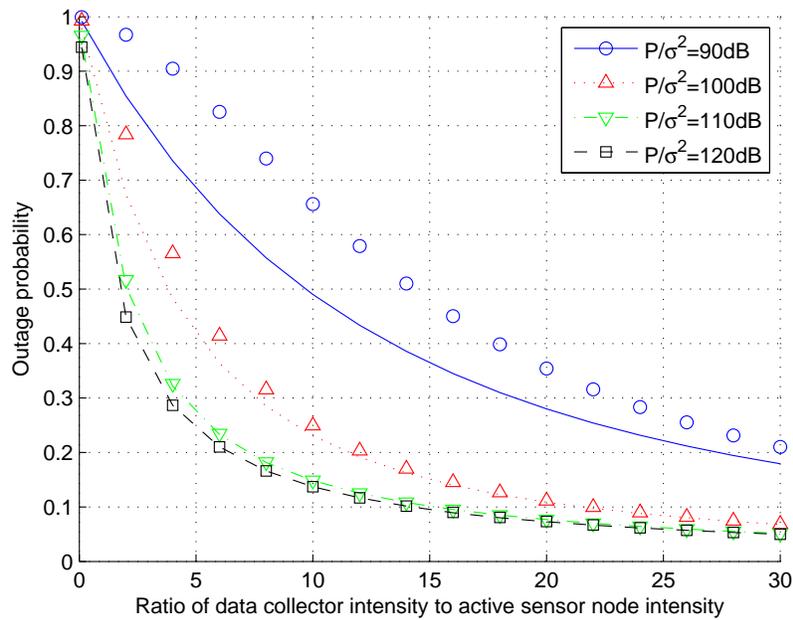}
\caption{Effect of noise power on outage probability ($\alpha =
4$; $\beta_t = 0$ dB; lines and symbols represent the exact
performances and their approximations by
(\ref{EQN_Density_Rayleigh_Sinr}), respectively))}
\label{FIG_EffectNoise_Rayleigh}
\end{figure}

\begin{figure}[p]
\centering
\includegraphics[width=12cm]{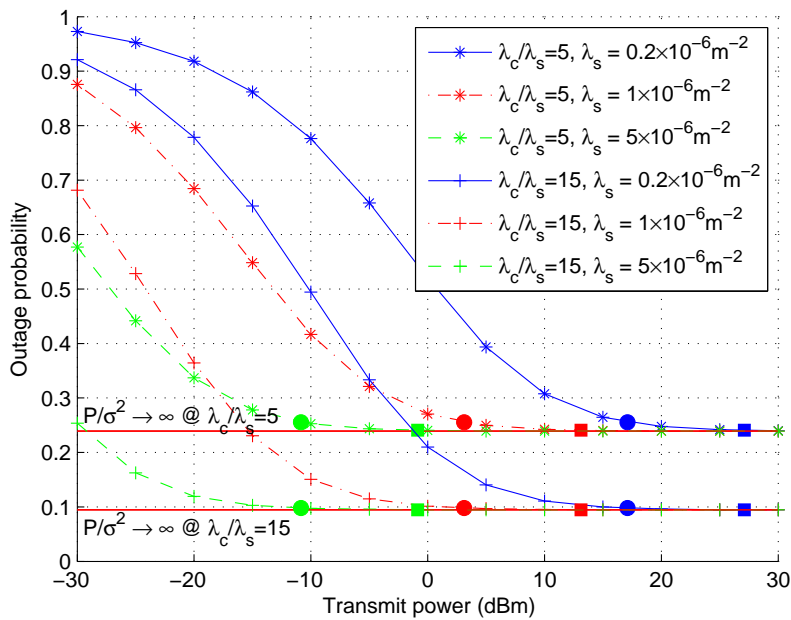}
\caption{Effect of transmit power on outage probability ($\alpha =
4$; $\beta_t = 0$ dB; solid circles and squares represent the
transmit power values by the design of (\ref{EQN_TxPw_Design})
with $c=0.1$ and $c=0.01$, respectively)}
\label{FIG_EffectTxPower_Rayleigh}
\end{figure}

\begin{figure}[p]
\centering
\includegraphics[width=12cm]{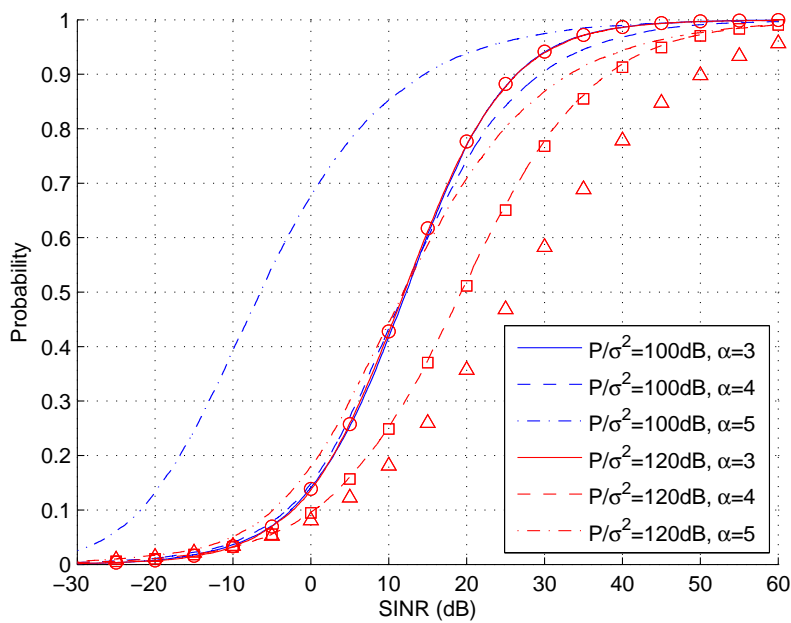}
\caption{CDF of SINR for Rayleigh fading channels according to the
pathloss exponents (lines - simulation results; circles, squares
and triangles represent the case of $P/\sigma^2 \rightarrow
\infty$ for $\alpha=3,4,5$, respectively)}
\label{FIG_SinrCdf_Rayleigh_PathlossExp}
\end{figure}

\begin{figure}[p]
\centering
\includegraphics[width=12cm]{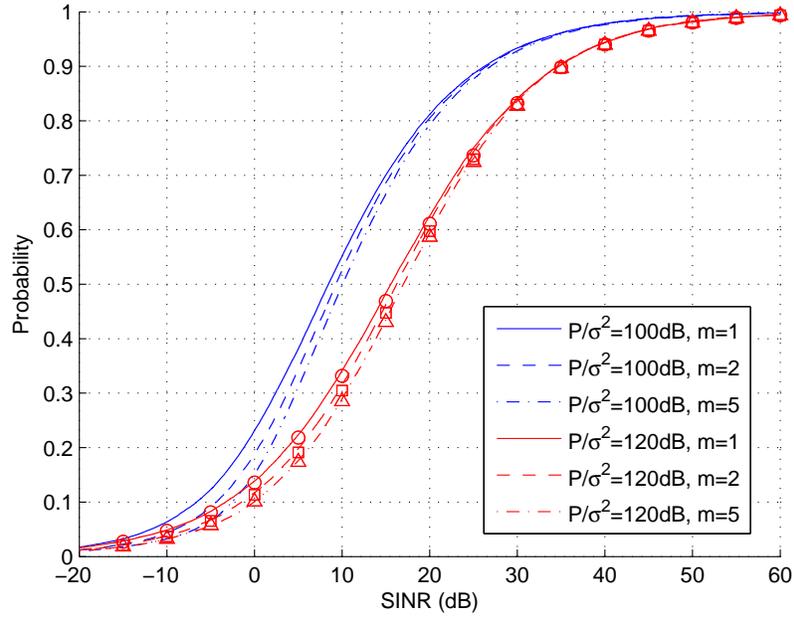}
\caption{CDF of SINR for Nakagami-$m$ fading channels ($\alpha =
4$; lines - simulation results; symbols - analysis results when
$P/\sigma^2 \rightarrow \infty$)} \label{FIG_SinrCdf_Nakagami}
\end{figure}

\begin{figure}[p]
\centering
\includegraphics[width=12cm]{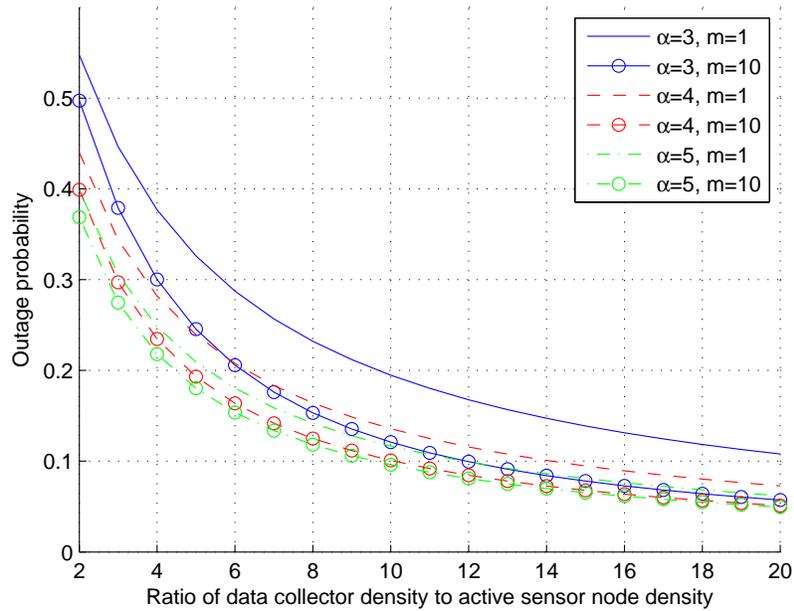}
\caption{Effect of wireless channels on outage probability
($\beta_t = 0$ dB; $P/\sigma^2 \rightarrow \infty$)}
\label{FIG_EffectChannel_Nakagami}
\end{figure}

\begin{figure}[p]
\centering
\includegraphics[width=12cm]{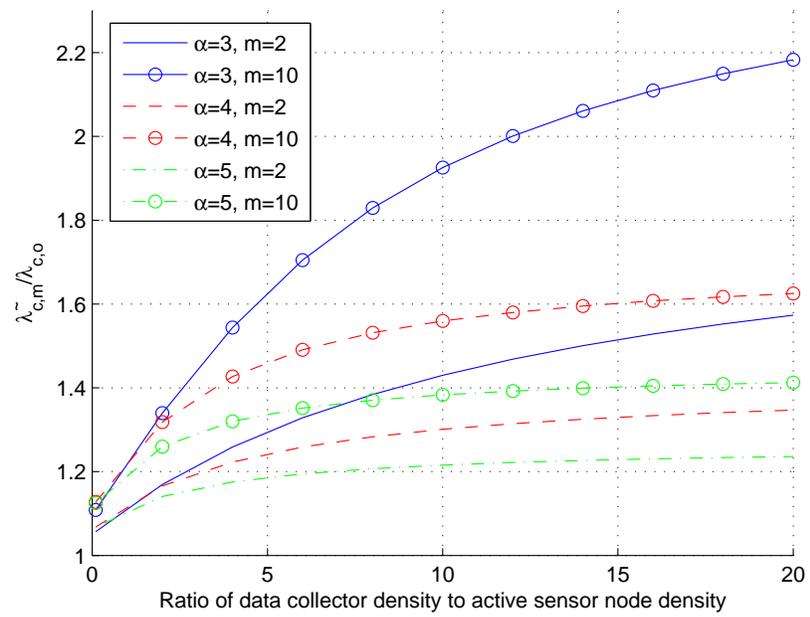}
\caption{Relative effect of wireless channels on the system
designs, compared to Rayleigh fading ($\beta_t = 0$ dB;
$P/\sigma^2 \rightarrow \infty$)}
\label{FIG_EffectChannel_Density}
\end{figure}

\end{document}